\documentclass[journal,onecolumn]{IEEEtran}
\IEEEoverridecommandlockouts

\usepackage[T1]{fontenc}
\usepackage[utf8]{inputenc}

\usepackage{cite}
\usepackage{amsmath,amssymb,amsfonts,amsthm,bbm,mathtools}

\usepackage{xcolor}
\usepackage{cleveref}
\usepackage{graphicx}
\usepackage{tikz}
\usetikzlibrary{spy}
\usetikzlibrary{positioning,calc}
\usetikzlibrary{decorations.pathmorphing,calc,shapes,shapes.geometric,patterns}
\usetikzlibrary{shapes.multipart}
\usetikzlibrary{arrows,positioning,calc,intersections}
\usetikzlibrary{datavisualization.formats.functions}
\usetikzlibrary{arrows, decorations.markings}
\usetikzlibrary{arrows.meta}

\newtheorem{theorem}{Theorem}
\newtheorem*{theorem*}{Theorem}

\newtheorem{definition}{Definition}

\newtheorem{corollary}[theorem]{Corollary}
\theoremstyle{remark}
\newtheorem{remark}{Remark}
\newtheorem*{remark*}{Remark}

\begin{document}

\IEEEoverridecommandlockouts
 \IEEEpubid{\makebox[\columnwidth]{\hfill} \hspace{\columnsep}\makebox[\columnwidth]{ }}

\title{Secure Deterministic Identification with Feedback and Sensing over State-Dependent Wiretap Channels\\
}
\author{
\IEEEauthorblockN{
\hspace{0.4cm}Vida Gholamian%
    \IEEEauthorrefmark{1}%
    \IEEEauthorrefmark{5}%
    ,
Yaning Zhao%
    \IEEEauthorrefmark{1}%
    \IEEEauthorrefmark{5}%
    ,
Christian Deppe%
    \IEEEauthorrefmark{1}%
    \IEEEauthorrefmark{5}%
    ,
Holger Boche%
    \IEEEauthorrefmark{2}%
    \IEEEauthorrefmark{4}%
    \IEEEauthorrefmark{5}%
    \IEEEauthorrefmark{6}%
    \IEEEauthorrefmark{7}%
    ,
Rafael F. Schaefer%
    \IEEEauthorrefmark{3}%
    \IEEEauthorrefmark{4}%
    \IEEEauthorrefmark{5}%
}

\vspace{0.5em}

\IEEEauthorblockA{\IEEEauthorrefmark{1}
Technische Universität Braunschweig, Institute for Communications Technology, Brunswick, Germany}

\IEEEauthorblockA{\IEEEauthorrefmark{2}
Technische Universität Munich, Chair of Theoretical Information Technology, Munich, Germany}

\IEEEauthorblockA{\IEEEauthorrefmark{3}
Technische Universität Dresden, Chair of Information Theory and Machine Learning, Dresden, Germany}

\IEEEauthorblockA{\IEEEauthorrefmark{4}
Cluster of Excellence “Centre for Tactile Internet with Human-in-the-Loop” (CeTI) at Technische Universität Dresden}

\IEEEauthorblockA{\IEEEauthorrefmark{5}
BMFTR Transfer Hub 6G-life}

\IEEEauthorblockA{\IEEEauthorrefmark{6}
Munich Center for Quantum Science and Technology (MCQST)}

\IEEEauthorblockA{\IEEEauthorrefmark{7}
Munich Quantum Valley (MQV)}

\IEEEauthorblockA{
Email: \{v.gholamian-sefiddarboni, yaning.zhao, christian.deppe\}@tu-bs.de,
boche@tum.de, rafael.schaefer@tu-dresden.de
}

\vspace{-0.8cm}
}

\maketitle

\thispagestyle{empty}
\pagestyle{empty}
\pagenumbering{gobble}

\begin{abstract}
Identification differs fundamentally from classical
message transmission: rather than reconstructing
a transmitted message, the receiver only determines
whether a particular identity was sent. This allows
the number of identifiable messages to grow doubly
exponentially with the blocklength, provided that
sufficient randomness is available. In deterministic
identification, however, the encoder has no access
to local randomization, making the generation of
common randomness a central challenge.

In this paper, we investigate secure deterministic
identification over state-dependent discrete
memoryless wiretap channels with private noiseless
feedback and a sensing distortion constraint.
Our framework combines three requirements:
reliable identification at the legitimate receiver,
confidentiality of the transmitted identity against
an eavesdropper, and estimation of the channel
state at the transmitter. The encoder is fully
deterministic and may depend only on the identity
and previously received feedback. We impose a
symbol-wise input restriction that guarantees
the prescribed sensing distortion constraint.

We establish conditions under which secure
deterministic identification achieves the same
capacity as non-secure identification with feedback.
Specifically, if the admissible main channel has
positive ordinary transmission capacity and the
conditional entropy of the legitimate channel
output given the eavesdropper's observation and
a fixed admissible input is positive, secrecy
does not reduce the identification capacity.
Importantly, this result does not require
positive classical wiretap secrecy capacity.

The main achievability construction exploits
noiseless feedback to generate common randomness
for identification and to extract a secret key
through privacy amplification. The key protects
a short identification tag using one-time-pad
encryption, while reliable transmission of the
encrypted tag requires only a sublinear number
of additional channel uses. We also develop an
alternative construction in which
feedback-generated randomness replaces the
local randomization of a stochastic wiretap
encoder.

Furthermore, we establish that the secure
identification capacity is zero whenever the
non-secure identification capacity vanishes.
We also identify a class of channels for which
the full non-secure identification capacity
is achievable even when the private feedback
entropy is zero, and illustrate this result
with a binary state-dependent wiretap channel.

Our results demonstrate that feedback can
provide the randomness required for both
identification and secrecy without local
encoder randomization, while simultaneously
satisfying sensing constraints.
\end{abstract}


\pagestyle{plain}

\section{Introduction}

    Identification (ID) via discrete memoryless channels (DMCs), introduced by Ahlswede and Dueck~\cite{AhlswedeDueck1989}, provides a fundamentally different communication paradigm compared to classical message transmission over DMCs~\cite{6773024}. Instead of requiring the receiver to reconstruct the transmitted message, the receiver performs a binary hypothesis test for each candidate identity. When local randomness is available at the encoder, the number of identifiable messages can grow double-exponentially with the blocklength. Hence, the ID rate is measured as $\frac{1}{n}\log\log N$, rather than the classical transmission rate $\frac{1}{n}\log N$. This setting is commonly referred to as randomized ID. In contrast, without local encoder randomness, deterministic ID supports only exponentially many identifiable messages, and its rate behaves similarly to the classical transmission rate~\cite{salariseddigh2021deterministic}.
    
    The central feature of ID is that randomization is so important that feedback can enable double exponential rate even in the absence of local randomness~\cite{wiese2022identification,labidi2024joint, zhao2025identification}, by making the randomness of the stochastic
    channel output available to both the transmitter and the
    legitimate receiver.
    This common randomness determines the ID capacity and forms the basis of achievability schemes~\cite{ahlswedeCai1999,verdu2008information,shemuel2022feedback}. This suggests that the randomness required not only for ID but also for secrecy need not be produced at the encoder itself, but can instead be extracted from the communication process. The main insight of this work is that feedback-induced randomness is sufficient to achieve full ID performance, making encoder randomization unnecessary even in the presence of a wiretapper.
    
    This observation is particularly relevant for practical resource-constrained systems, such as IoT sensors, implantable medical devices, backscatter nodes, and molecular or nanoscale communication systems, where generating reliable local randomness may be costly, infeasible, or insecure. Insufficient entropy can lead to weak cryptographic keys and serious security vulnerabilities~\cite{heninger2012mining}, while randomness sources are often limited by hardware and environmental constraints~\cite{capkun2025security}. In extreme molecular or nanoscale scenarios, transmitters may even lack mechanisms for controlled randomization~\cite{nakano2013molecular,sabu2020analysis}.
    This limitation makes the considered framework particularly relevant to molecular communication, where the combination of post-Shannon communication and security has been identified as an interesting open problem and a promising direction for future research~\cite{haselmayr2019integration}.
    The scenario considered is motivated by integrated sensing and communication (ISAC), where sensing and data transmission share the same physical resources~\cite{fettweis20216g}. In such systems, nodes must observe a dynamically changing environment while simultaneously exchanging information, which calls for a joint design rather than a separation-based approach. Existing ISAC studies focus mainly on waveform design and signal processing aspects~\cite{sturm2011waveform,bliss2014cooperative,bica2015opportunistic,huang2015radar}. 
    
    In an ISAC scenario, an ID message represents an
   event or operational label, such as the occurrence of an
   alarm, a device status, or a monitored operating condition,
    rather than a payload that must be fully reconstructed.
    The receiver therefore tests whether a particular event of
    interest has occurred while the same channel observations are
    used for state sensing.
    
    Although information-theoretic studies have characterized sensing and communication-tradeoffs for conventional message transmission~\cite{2025exploring,ahmadipour2022information,chang2023rate}, secure sensing-constrained ID with feedback in
    the absence of local encoder randomness remains unexplored.
    The inclusion of a wiretapper is motivated by the broadcast nature of wireless systems. In particular, identifying transmitted messages may reveal sensitive information about the system state or operational behavior. This issue is especially critical in sensing-integrated systems, where communication signals are closely coupled to the environment. Incorporating a wiretapper therefore captures fundamental security constraints and enables the characterization of reliable and confidential ID under sensing constraints.
    
    To address this gap, we adopt a state-dependent channel model with sensing constraints and focus on secure ID, which is well suited for control, monitoring, and alarm-type communication. 
    We further incorporate noiseless feedback and security constraints in the presence of this wiretapper. 
    In particular, this work builds on~\cite{EW26} which considers secure
    ID under sensing constraints with randomized
    encoding.

The remainder of the paper is organized as
follows. Section II introduces the system
model, sensing criterion, and main capacity
result. Section III establishes achievability
using feedback-generated common randomness
and one-time-pad protection of the
identification tag. Section IV presents an
alternative achievability construction based
on stochastic wiretap encoding. Section V
concludes the paper.


\section{System Models and Main Results} \label{sec:Preli}
    Let $(\mathcal{X} \times \mathcal{S},P_S, W_S(y|x,s), V_S(z|x,s), \mathcal{Y}, \mathcal{Z})$ denote a state-dependent discrete memoryless wiretap channel, where $\mathcal{X}$ is a finite input alphabet, $\mathcal{Y}$ and $\mathcal{Z}$ are finite output alphabets corresponding to the legitimate receiver and the wiretapper, respectively, and $\mathcal{S}$ is a finite state alphabet. The channel is characterized by the conditional probability mass functions $W_S(y|x,s)$ on $\mathcal{Y}$ and $V_S(z|x,s)$ on $\mathcal{Z}$ and we denote the state-averaged channels as $W = P_S W_S$ and $V = P_S V_S$.
    Additionally, for an input distribution $P_X$ we denote the output distributions as $P_Y = P_X W = P_X P_S W_S$ and $P_Z = P_X V = P_X P_S V_S$. 
    Let $P_{YZ|X,S}$ denote the joint channel law
    with marginals $W_S$ and $V_S$. The corresponding
    state-averaged joint channel is
    \begin{equation}
    \label{eq:joint_channel}
    P_{YZ|X}(y,z|x)
    =\sum_{s\in\mathcal S}P_S(s)
    P_{YZ|X,S}(y,z|x,s).
    \end{equation}

Throughout the paper, all logarithms are
to base two unless explicitly stated
otherwise. Accordingly, all entropies
and information rates are measured in bits.
The notation $\exp(x)$ denotes $e^x$.

    The wiretap channel is jointly memoryless.
For any fixed input sequence
$x^n\in\mathcal X^n$ and state sequence
$s^n\in\mathcal S^n$, its joint
transition law satisfies
\[
P_{YZ|X,S}^{(n)}(y^n,z^n|x^n,s^n)
=
\prod_{t=1}^n
P_{YZ|X,S}(y_t,z_t|x_t,s_t).
\]
In particular, the corresponding
marginal transition laws satisfy
\[
\begin{aligned}
W_S^n(y^n|x^n,s^n)
&=\prod_{t=1}^n W_S(y_t|x_t,s_t),\\
V_S^n(z^n|x^n,s^n)
&=\prod_{t=1}^n V_S(z_t|x_t,s_t).
\end{aligned}
\]
  The state sequence $(S_t)_{t=1}^n$ is i.i.d.\
according to $P_S$. The state sequence is
not provided as additional side information
to the transmitter, the legitimate receiver,
or the wiretapper. However, the terminals
may infer information about the states
from their respective channel observations.
At each time $t$, the current state $S_t$
is independent of the identification message
and all past channel inputs and outputs.

    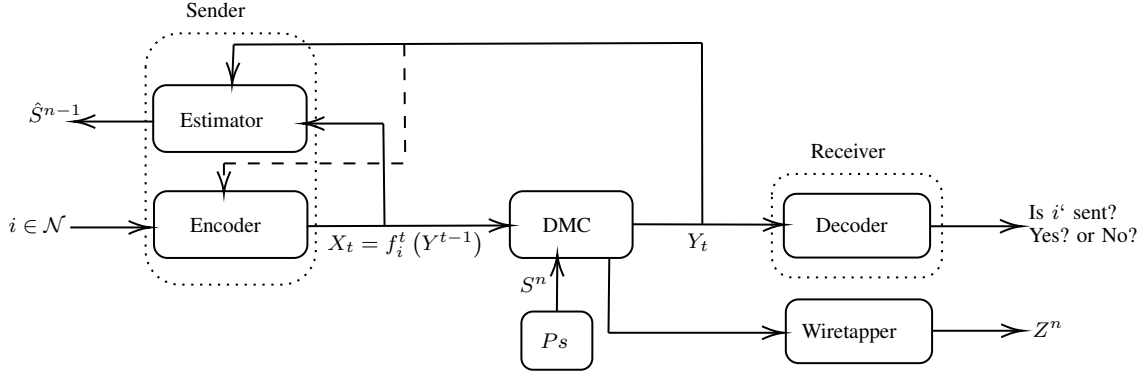
\begin{figure}
        \centering
        \scalebox{1}{

\tikzset{every picture/.style={line width=0.75pt}} 

\begin{tikzpicture}[x=0.75pt,y=0.75pt,yscale=-1,xscale=1]

\draw   (82,61.7) .. controls (82,58) and (85,55) .. (88.7,55) -- (152.3,55) .. controls (156,55) and (159,58) .. (159,61.7) -- (159,81.8) .. controls (159,85.5) and (156,88.5) .. (152.3,88.5) -- (88.7,88.5) .. controls (85,88.5) and (82,85.5) .. (82,81.8) -- cycle ;
\draw   (82.4,114.9) .. controls (82.4,111.2) and (85.4,108.2) .. (89.1,108.2) -- (152.7,108.2) .. controls (156.4,108.2) and (159.4,111.2) .. (159.4,114.9) -- (159.4,135) .. controls (159.4,138.7) and (156.4,141.7) .. (152.7,141.7) -- (89.1,141.7) .. controls (85.4,141.7) and (82.4,138.7) .. (82.4,135) -- cycle ;
\draw   (261,114.9) .. controls (261,111.2) and (264,108.2) .. (267.7,108.2) -- (315.9,108.2) .. controls (319.6,108.2) and (322.6,111.2) .. (322.6,114.9) -- (322.6,135) .. controls (322.6,138.7) and (319.6,141.7) .. (315.9,141.7) -- (267.7,141.7) .. controls (264,141.7) and (261,138.7) .. (261,135) -- cycle ;
\draw   (398.5,115.9) .. controls (398.5,112.37) and (401.37,109.5) .. (404.9,109.5) -- (465.5,109.5) .. controls (469.03,109.5) and (471.9,112.37) .. (471.9,115.9) -- (471.9,135.1) .. controls (471.9,138.63) and (469.03,141.5) .. (465.5,141.5) -- (404.9,141.5) .. controls (401.37,141.5) and (398.5,138.63) .. (398.5,135.1) -- cycle ;
\draw    (159,125.75) -- (259.4,125.41) ;
\draw [shift={(261.4,125.4)}, rotate = 179.8] [color={rgb, 255:red, 0; green, 0; blue, 0 }  ][line width=0.75]    (10.93,-3.29) .. controls (6.95,-1.4) and (3.31,-0.3) .. (0,0) .. controls (3.31,0.3) and (6.95,1.4) .. (10.93,3.29)   ;
\draw    (322.4,124.8) -- (397,124.99) ;
\draw [shift={(399,125)}, rotate = 180.15] [color={rgb, 255:red, 0; green, 0; blue, 0 }  ][line width=0.75]    (10.93,-3.29) .. controls (6.95,-1.4) and (3.31,-0.3) .. (0,0) .. controls (3.31,0.3) and (6.95,1.4) .. (10.93,3.29)   ;
\draw    (357.6,125) -- (357.4,33.8) ;
\draw    (357.4,33.8) -- (122.4,34.6) ;
\draw  [dash pattern={on 0.84pt off 2.51pt}] (78.2,46.08) .. controls (78.2,36.65) and (85.85,29) .. (95.28,29) -- (146.52,29) .. controls (155.95,29) and (163.6,36.65) .. (163.6,46.08) -- (163.6,138.02) .. controls (163.6,147.45) and (155.95,155.1) .. (146.52,155.1) -- (95.28,155.1) .. controls (85.85,155.1) and (78.2,147.45) .. (78.2,138.02) -- cycle ;
\draw    (122.4,34.6) -- (121.68,52.6) ;
\draw [shift={(121.6,54.6)}, rotate = 272.29] [color={rgb, 255:red, 0; green, 0; blue, 0 }  ][line width=0.75]    (10.93,-3.29) .. controls (6.95,-1.4) and (3.31,-0.3) .. (0,0) .. controls (3.31,0.3) and (6.95,1.4) .. (10.93,3.29)   ;
\draw   (266.8,174.44) .. controls (266.8,171.21) and (269.41,168.6) .. (272.64,168.6) -- (296.56,168.6) .. controls (299.79,168.6) and (302.4,171.21) .. (302.4,174.44) -- (302.4,191.96) .. controls (302.4,195.19) and (299.79,197.8) .. (296.56,197.8) -- (272.64,197.8) .. controls (269.41,197.8) and (266.8,195.19) .. (266.8,191.96) -- cycle ;
\draw    (284.6,168.4) -- (284.82,142) ;
\draw [shift={(284.83,140)}, rotate = 90.47] [color={rgb, 255:red, 0; green, 0; blue, 0 }  ][line width=0.75]    (10.93,-3.29) .. controls (6.95,-1.4) and (3.31,-0.3) .. (0,0) .. controls (3.31,0.3) and (6.95,1.4) .. (10.93,3.29)   ;
\draw    (472.3,125.8) -- (515.9,125.8) ;
\draw [shift={(517.9,125.8)}, rotate = 180] [color={rgb, 255:red, 0; green, 0; blue, 0 }  ][line width=0.75]    (10.93,-3.29) .. controls (6.95,-1.4) and (3.31,-0.3) .. (0,0) .. controls (3.31,0.3) and (6.95,1.4) .. (10.93,3.29)   ;
\draw    (197.5,74.25) -- (159,74.25) ;
\draw [shift={(157,74.25)}, rotate = 360] [color={rgb, 255:red, 0; green, 0; blue, 0 }  ][line width=0.75]    (10.93,-3.29) .. controls (6.95,-1.4) and (3.31,-0.3) .. (0,0) .. controls (3.31,0.3) and (6.95,1.4) .. (10.93,3.29)   ;
\draw    (197.5,74.25) -- (198.1,125.3) ;
\draw  [dash pattern={on 4.5pt off 4.5pt}]  (117.5,94.25) -- (117.5,107.25) ;
\draw [shift={(117.5,109.25)}, rotate = 270] [color={rgb, 255:red, 0; green, 0; blue, 0 }  ][line width=0.75]    (10.93,-3.29) .. controls (6.95,-1.4) and (3.31,-0.3) .. (0,0) .. controls (3.31,0.3) and (6.95,1.4) .. (10.93,3.29)   ;
\draw  [dash pattern={on 4.5pt off 4.5pt}]  (117.5,94.25) -- (208.5,94.25) ;
\draw  [dash pattern={on 4.5pt off 4.5pt}]  (208,34.5) -- (208.5,94.25) ;
\draw    (82.5,73.25) -- (44,73.25) ;
\draw [shift={(42,73.25)}, rotate = 360] [color={rgb, 255:red, 0; green, 0; blue, 0 }  ][line width=0.75]    (10.93,-3.29) .. controls (6.95,-1.4) and (3.31,-0.3) .. (0,0) .. controls (3.31,0.3) and (6.95,1.4) .. (10.93,3.29)   ;
\draw    (40.4,126.4) -- (78.6,126.4) ;
\draw [shift={(80.6,126.4)}, rotate = 180] [color={rgb, 255:red, 0; green, 0; blue, 0 }  ][line width=0.75]    (10.93,-3.29) .. controls (6.95,-1.4) and (3.31,-0.3) .. (0,0) .. controls (3.31,0.3) and (6.95,1.4) .. (10.93,3.29)   ;
\draw    (472.9,178.24) -- (516.5,178.24) ;
\draw [shift={(518.5,178.24)}, rotate = 180] [color={rgb, 255:red, 0; green, 0; blue, 0 }  ][line width=0.75]    (10.93,-3.29) .. controls (6.95,-1.4) and (3.31,-0.3) .. (0,0) .. controls (3.31,0.3) and (6.95,1.4) .. (10.93,3.29)   ;
\draw    (310.5,179.25) -- (310.9,141.8) ;
\draw   (399.5,168.9) .. controls (399.5,165.37) and (402.37,162.5) .. (405.9,162.5) -- (466.5,162.5) .. controls (470.03,162.5) and (472.9,165.37) .. (472.9,168.9) -- (472.9,188.1) .. controls (472.9,191.63) and (470.03,194.5) .. (466.5,194.5) -- (405.9,194.5) .. controls (402.37,194.5) and (399.5,191.63) .. (399.5,188.1) -- cycle ;
\draw    (310.5,179.25) -- (396.5,179.25) ;
\draw [shift={(398.5,179.25)}, rotate = 180] [color={rgb, 255:red, 0; green, 0; blue, 0 }  ][line width=0.75]    (10.93,-3.29) .. controls (6.95,-1.4) and (3.31,-0.3) .. (0,0) .. controls (3.31,0.3) and (6.95,1.4) .. (10.93,3.29)   ;
\draw  [dash pattern={on 0.84pt off 2.51pt}] (392.5,109.92) .. controls (392.5,104.18) and (397.15,99.54) .. (402.89,99.54) -- (467.51,99.54) .. controls (473.25,99.54) and (477.9,104.18) .. (477.9,109.92) -- (477.9,141.08) .. controls (477.9,146.82) and (473.25,151.47) .. (467.51,151.47) -- (402.89,151.47) .. controls (397.15,151.47) and (392.5,146.82) .. (392.5,141.08) -- cycle ;

\draw (94,67) node [anchor=north west][inner sep=0.75pt]  [font=\footnotesize] [align=left] {Estimator};
\draw (99,120) node [anchor=north west][inner sep=0.75pt]  [font=\footnotesize] [align=left] {Encoder};
\draw (276,120) node [anchor=north west][inner sep=0.75pt]  [font=\footnotesize] [align=left] {DMC};
\draw (413,120) node [anchor=north west][inner sep=0.75pt]  [font=\footnotesize] [align=left] {Decoder};
\draw (274.9,178) node [anchor=north west][inner sep=0.75pt]  [font=\footnotesize]  {$Ps$};
\draw (168,127) node [anchor=north west][inner sep=0.75pt]  [font=\footnotesize]  {$X_{t} =f_{i}^{t}\left( Y^{t-1}\right)$};
\draw (264.56,150) node [anchor=north west][inner sep=0.75pt]  [font=\footnotesize]  {$S^{n}$};
\draw (18,62) node [anchor=north west][inner sep=0.75pt]  [font=\footnotesize]  {$\hat{S}^{n-1}$};
\draw (348.56,128) node [anchor=north west][inner sep=0.75pt]  [font=\footnotesize]  {$Y_{t}$};
\draw (97,12) node [anchor=north west][inner sep=0.75pt]  [font=\footnotesize] [align=left] {Sender};
\draw (520,113) node [anchor=north west][inner sep=0.75pt]  [font=\footnotesize] [align=left] {Is $\displaystyle i`$ sent?\\Yes? or No?};
\draw (8,120) node [anchor=north west][inner sep=0.75pt]  [font=\footnotesize] [align=left] {$\displaystyle i\in \mathcal{N}$};
\draw (406,173) node [anchor=north west][inner sep=0.75pt]  [font=\footnotesize] [align=left] {Wiretapper};
\draw (522,172) node [anchor=north west][inner sep=0.75pt]  [font=\footnotesize] [align=left] {$\displaystyle Z^{n}$};
\draw (410.5,83) node [anchor=north west][inner sep=0.75pt]  [font=\footnotesize] [align=left] {Receiver};

\end{tikzpicture}

}
        \caption{State-dependent channel with noiseless feedback in the presence of a wiretapper.}
        \label{fig:System}
    \end{figure}

    The system can be interpreted as consisting of a sender and a receiver. The sender comprises an encoder and a state estimator. For each identity $i \in \{1,\ldots,N\}$, the encoder generates channel inputs according to $x_t = f_t^i(y^{t-1})$, where $y^{t-1} \in \mathcal{Y}^{t-1}$ denotes the delayed noiseless feedback output. The noiseless feedback link from the legitimate receiver to the transmitter is assumed to be private; hence, the wiretapper does not observe the feedback sequence \(Y^{t-1}\). 
     The estimator produces a state estimate sequence $\hat{s}^n \in \mathcal{S}^n$ based on the channel input and the observed feedback.
   The legitimate receiver observes $y^n$ and performs a hypothesis test to decide whether a given identity was transmitted.
  The state estimator is located at the
transmitter, which has access to its own
channel inputs and, through noiseless
feedback, to the corresponding channel
outputs. The distortion constraint
therefore measures the quality of state
estimation at the transmitter.
    
    In addition to reliable ID at the legitimate receiver and the sensing distortion constraint, we impose a secrecy requirement with respect to the wiretapper’s observation. In the deterministic ID setting, security is formulated in terms of the wiretapper’s inability to distinguish between different ID messages.
    More precisely, for any pair of distinct messages, the corresponding induced output distributions at the wiretapper must be asymptotically indistinguishable. 
This pairwise indistinguishability criterion
provides an ID-oriented total-variation
secrecy requirement. It ensures that the
wiretapper's distinguishing advantage
between any two identification messages
vanishes asymptotically, uniformly over
all pairs of identities. This criterion
is closely related to total-variation
formulations of semantic security for
wiretap channels~\cite{bellare2012semantic,yang2019wiretap}.

  A secure deterministic ID feedback code for the state-dependent wiretap channel with feedback and sensing is defined as follows.

\begin{definition}
\label{def:1}
An $(n,N,\lambda)$ secure deterministic ID feedback
code with $\lambda<1$ for $(W_S,V_S)$ consists
of pairs
\[
\left\{(f_i,\mathcal D_i):i=1,\ldots,N\right\},
\]
where $f_i=(f_i^1,\ldots,f_i^n)$ denotes the
deterministic feedback encoding strategy for
identity $i$, with
\[
f_i^1\in\mathcal X_D,
\qquad
f_i^t:\mathcal Y^{t-1}\to\mathcal X_D,
\quad t=2,\ldots,n,
\]
and $\mathcal D_i\subseteq\mathcal Y^n$ is the
decoding region associated with identity $i$.

For each identity $i$, let $P_i$ denote the
joint output distribution induced by the
deterministic feedback encoding strategy and
the state-averaged joint channel:
\begin{equation}
P_i(y^n,z^n)
=
\prod_{t=1}^{n}
P_{YZ|X}
\left(
y_t,z_t\mid f_i^t(y^{t-1})
\right).
\end{equation}
Here, $f_i^1(y^0)$ is understood as the
constant input $f_i^1$.
Let $P_{Y^n|i}$ and $P_{Z^n|i}$ denote
the corresponding marginal distributions.

The code satisfies the following reliability
and secrecy conditions for all
$i\in\{1,\ldots,N\}$ and all $j\neq i$:
\begin{align}
    \label{eq:def_det_first_kind}
    P_{Y^n|i}\bigl(\mathcal D_i^c\bigr)
    &\leq \lambda,
    \\
    \label{eq:def_det_second_kind}
    P_{Y^n|i}\bigl(\mathcal D_j\bigr)
    &\leq \lambda,
    \\
    \label{eq:def_det_security}
    \left\|
        P_{Z^n|i}-P_{Z^n|j}
    \right\|_{\mathrm{TV}}
    &\leq \lambda.
\end{align}
\end{definition}

    \begin{remark}
 If the channel state is known to the legitimate
receiver, the decoding regions may additionally
depend on the state sequence, i.e.,
$\mathcal{D}_i(s^n)$.
Such additional information cannot decrease
the achievable secure identification capacity.
However, the capacity characterization in
this paper concerns the model in which the
channel state is unknown to the receiver.
In general, equality with the capacity of
the state-informed receiver model cannot
be asserted without additional assumptions.
    \end{remark}
       
    The availability of noiseless feedback plays a dual role in the
    considered model. On the one hand, it enables the generation of
    common randomness between the transmitter and the legitimate
    receiver, which is essential for ID without requiring local encoder
    randomization. On the other hand, it provides the observations
    required for state estimation. Thus, feedback supports both the
    ID and sensing tasks.
    
  For state estimation, we employ the optimal
symbol-wise estimator with respect to the
distortion function $d$ defined below.
For each $(x,y)\in\mathcal X\times\mathcal Y$,
let
\begin{equation}
h^*(x,y)\in
\operatorname*{arg\,min}_{\hat s\in\mathcal S}
\sum_{s\in\mathcal S}
P_S(s)W_S(y|x,s)d(s,\hat s).
\end{equation}
The state estimate at time $t$ is then
given by
\begin{equation}
\hat S_t=h^*(X_t,Y_t),
\quad t=1,\ldots,n.
\end{equation}
    
    Since the channel inputs $X_t$ are generated via deterministic encoding functions that depend on past outputs, the joint distribution of $(X_t,Y_t,S_t)$ is induced by the feedback strategy and the channel law.
    
   \begin{remark}
Since the state sequence is i.i.d. and the
channel is memoryless, the posterior
distribution of the state sequence given
the input and output sequences factorizes as
\[
P_{S^n|X^n,Y^n}(s^n|x^n,y^n)
=
\prod_{t=1}^n
P_{S|X,Y}(s_t|x_t,y_t)
\]
for all $(x^n,y^n)$ with positive probability.
Consequently, under the additive distortion
measure considered here, the optimal
symbol-wise estimator $h^*$ also minimizes
the expected distortion among all estimators
based on the entire sequences $(X^n,Y^n)$.
Thus, noncausal sequence estimation cannot
improve the minimum achievable average
distortion.
\end{remark}
    
    Here, $d:\mathcal S \times \mathcal S \to [0,\infty)$ denotes the distortion function, and the expectation is taken with respect to the joint distribution induced by the deterministic encoding strategy and the channel. We further define
    \begin{align}
    d^*_S(x) &:= \mathbb{E}_{SY}\big[d(S,h^*(X,Y))\mid X=x\big],
    \end{align}
   We define the set of admissible input symbols as
\[
\mathcal{X}_D
:=
\left\{
x\in\mathcal{X}: d_S^*(x)\le D
\right\}.
\]
Throughout the paper, we assume that
$\mathcal{X}_D\neq\emptyset$, or equivalently,
\[
D\ge \min_{x\in\mathcal X}d_S^*(x).
\]
    Hence, for any \(x\in\mathcal X_D\),
    \begin{equation}
        \mathbb{E}\left[
        d(S_t,\hat S_t)\mid X_t=x
        \right]
        =d_S^*(x)\le D,
        \qquad x\in\mathcal X_D.
    \end{equation}
    The average distortion over a block of length $n$ is defined as
    \begin{equation}
    \bar d^{(n)} := \frac{1}{n}\sum_{t=1}^{n} \mathbb{E}\!\left[d(S_t,\hat S_t)\right],
    \end{equation}
   In the present model, we impose the stronger
requirement that every feedback encoding
function takes values in $\mathcal X_D$.
This symbol-wise admissibility condition
implies the average distortion constraint,
since
\[
\mathbb{E}[d(S_t,\hat S_t)]
=
\sum_{x\in\mathcal X_D}
P_{X_t}(x)d_S^*(x)
\le D
\]
for every $t=1,\ldots,n$. Consequently,
$\bar d^{(n)}\le D$.

We emphasize that this restriction is
stronger than imposing only a block-average
distortion constraint. The capacity results
derived below refer to the class of codes
whose channel inputs are restricted to
$\mathcal X_D$.

    \begin{remark}
    In ISAC and IoT applications, this constraint quantifies the
    required accuracy in estimating an environmental or operational
    state while the same channel resources are simultaneously used
    for ID.
    \end{remark}
  
    In the following, we define the secure
deterministic identification capacity for
the restricted class of feedback codes
introduced in Definition~\ref{def:1}.

More precisely, for every identity $i$,
every time instant $t$, and every possible
feedback realization $y^{t-1}$, we require
\[
f_i^t(y^{t-1})\in\mathcal X_D.
\]
This input-support restriction guarantees
that the expected distortion does not
exceed $D$ at any time instant and,
consequently, that
\[
\bar d^{(n)}\le D.
\]

However, restricting every channel input
to $\mathcal X_D$ is generally stronger
than imposing only a block-average
distortion constraint.

The capacity characterization established
below therefore applies to the restricted
class of codes defined above. It does not
necessarily characterize the capacity
under a general block-average distortion
constraint.
    
    \begin{definition}\label{Def:SID-D}
    \begin{enumerate}
        \item The rate of a secure deterministic $(n,N,\lambda)$
        ID feedback code for $(W_S,V_S)$ is defined as
        $R = \frac{\log\log N}{n}$.
    
        \item An ID rate--distortion pair $(R,D)$ for
        $(W_S,V_S)$ is said to be achievable if,
        for every $\lambda\in(0,1)$, there exists
        $n_0(\lambda)$ such that for every
        $n\geq n_0(\lambda)$, an
        $(n,N,\lambda)$ secure deterministic ID
        feedback code exists with
        $\frac{1}{n}\log\log N\geq R,$
        where every encoding function takes values
in the admissible input alphabet
$\mathcal X_D$ defined above.

    
        \item The secure deterministic ID
        capacity--distortion function
        $C_{\mathrm{SID}}^d(D,W_S,V_S)$ is defined as
\[
\sup\bigl(\{0\}\cup
\{R>0:(R,D)\text{ is achievable}\}\bigr).
\]
    \end{enumerate}
    \end{definition}
    
    The following theorem establishes bounds
    on the secure deterministic identification
    capacity with private noiseless feedback
    under the sensing distortion constraint.
    It also identifies conditions under which
    the secure capacity coincides with the
    non-secure identification capacity or
    vanishes. The result is expressed in terms
    of the non-secure identification capacity,
    the ordinary transmission capacity of the
    admissible main channel, and the private
    conditional feedback entropy.

  Let $C_{\mathrm{ID}}^d(D,W_S)$ denote the
non-secure deterministic identification
capacity with noiseless feedback under
the sensing distortion constraint $D$.

As established in~\cite{labidi2024joint},
the entropy characterization requires
positive ordinary transmission capacity.
More precisely,
\begin{equation}
C_{\mathrm{ID}}^d(D,W_S)
=
\begin{cases}
\displaystyle
\max_{x\in\mathcal X_D} H(Y|X=x),
& C^d(D,W_S)>0,\\[2mm]
0,
& C^d(D,W_S)=0,
\end{cases}
\end{equation}
where $C^d(D,W_S)$ denotes the ordinary
transmission capacity of the admissible
main channel, defined below, and
\[
P_{Y|X=x}(y)
=
\sum_{s\in\mathcal S}
P_S(s)W_S(y|x,s).
\]
    For comparison, let $C_S^d(D,W_S,V_S)$
denote the classical strong-secrecy
transmission capacity of the state-averaged
wiretap channel without feedback, allowing
stochastic encoding
~\cite{Wyner1975,csiszar1978broadcast,
csiszar2011information}.
With the input alphabet restricted to
$\mathcal X_D$, it is given by
\begin{equation}
C_S^d(D,W_S,V_S)
=
\max_{\substack{P_{U,X}:\\
P_X(\mathcal X_D)=1}}
\left[I(U;Y)-I(U;Z)\right],
\end{equation}
where the maximization is over joint
distributions satisfying
$U-X-(Y,Z)$ and
\begin{equation}
P_{U,X,Y,Z}(u,x,y,z)
=
P_{U,X}(u,x)P_{YZ|X}(y,z|x),
\end{equation}
with $P_{YZ|X}$ defined
in~\eqref{eq:joint_channel}.

We emphasize that this quantity refers
to the classical wiretap channel without
feedback. A positive value of
$C_S^d(D,W_S,V_S)$ is not required for
our secure deterministic identification
achievability result.

    Theorem~\ref{theorm:secureCapaciy} provides bounds
    on the secure deterministic ID capacity over
    state-dependent wiretap channels with private
    noiseless feedback. The upper bound follows from
    the corresponding non-secure ID capacity
    \cite{labidi2024joint}. Moreover, the secure ID
    capacity coincides with the non-secure one whenever
    the ordinary transmission capacity and the private
    conditional feedback entropy are both positive,
    without requiring positive ordinary secrecy capacity.    
    Let $C^d(D,W_S)$ denote the ordinary transmission
    capacity of the state-averaged main channel
    $W=P_SW_S$, with input alphabet restricted to
    $\mathcal{X}_D$, i.e.,
    \(
    C^d(D,W_S)
    := \max_{P_X:\,P_X(\mathcal{X}_D)=1} I(X;Y).
    \)
 Furthermore, define
\[
H_{\mathrm{priv}}(D)
:= \max_{x\in\mathcal{X}_D} H(Y\mid Z,X=x).
\]
The conditional entropy is evaluated under
the state-averaged joint channel law
$P_{YZ|X}$ defined in~\eqref{eq:joint_channel}.
In particular, $H_{\mathrm{priv}}(D)$
depends on the joint distribution of
$(Y,Z)$ given $X$, and not only on the
marginal channels $W_S$ and $V_S$.
Throughout the paper, the joint channel
law is understood to be part of the
specified wiretap channel model and is
suppressed in the notation
$C_{\mathrm{SID}}^d(D,W_S,V_S)$.

    In the following theorem, all channel inputs
    are restricted to $\mathcal{X}_D$.
\begin{theorem}\label{theorm:secureCapaciy}
    The secure deterministic ID feedback capacity
    with sensing distortion constraint $D$ satisfies
    \begin{equation}
    0\leq C_{\mathrm{SID}}^d(D,W_S,V_S)
    \leq C_{\mathrm{ID}}^d(D,W_S).
    \end{equation}
    Moreover, the following equalities hold
    under the stated conditions:
    \begin{equation*}
    C_{\mathrm{SID}}^d(D,W_S,V_S)=
    \begin{cases}
    C_{\mathrm{ID}}^d(D,W_S),
    & C^d(D,W_S)>0 \text{ and } H_{\mathrm{priv}}(D)>0,\\[2mm]
    0,
    & C^d(D,W_S)=0 \text{ or } C_{\mathrm{ID}}^d(D,W_S)=0.
    \end{cases}
    \end{equation*}
\end{theorem}

The preceding theorem establishes
sufficient conditions for achieving
the full non-secure identification
capacity. However, the condition
$H_{\mathrm{priv}}(D)>0$ is not
necessary in general.

In particular, the regime
\[
C^d(D,W_S)>0,\qquad
C_{\mathrm{ID}}^d(D,W_S)>0,\qquad
H_{\mathrm{priv}}(D)=0
\]
is not completely characterized
by the preceding theorem.

The following corollary identifies
an additional class of wiretap
channels for which the full
non-secure identification capacity
is achievable, independently of
the value of $H_{\mathrm{priv}}(D)$.

\begin{corollary}
\label{cor:input_independent_wiretap}
Suppose that the state-averaged
wiretapper channel is independent
of the input, i.e., there exists
a probability distribution $Q_Z$
on $\mathcal Z$ such that
\[
P_{Z|X}(z|x)=Q_Z(z)
\]
for every $x\in\mathcal X_D$
and every $z\in\mathcal Z$.

Then the secure deterministic
identification capacity with
private noiseless feedback satisfies
\[
C_{\mathrm{SID}}^d(D,W_S,V_S)
=
C_{\mathrm{ID}}^d(D,W_S),
\]
regardless of the value of
$H_{\mathrm{priv}}(D)$.
\end{corollary}

\begin{proof}
Consider an arbitrary deterministic
identification feedback code for
the admissible main channel.

For identity $i$, the channel input
at time $t$ is given by
\[
X_t=f_i^t(Y^{t-1}).
\]

Since the channel is memoryless
and the state sequence is i.i.d.,
the current wiretapper output,
conditioned on the past observations,
has distribution
\begin{align}
&P_{Z_t|Y^{t-1},Z^{t-1},i}
(z_t|y^{t-1},z^{t-1},i)
\\
&\quad =
P_{Z|X}
(z_t|f_i^t(y^{t-1}))
\\
&\quad = Q_Z(z_t).
\end{align}

This equality holds for every
feedback history having positive
probability.

Consequently, averaging over the
legitimate receiver's past outputs
gives
\[
P_{Z_t|Z^{t-1},i}
(z_t|z^{t-1},i)
=
Q_Z(z_t).
\]

By the chain rule,
\[
P_{Z^n|i}(z^n)
=
\prod_{t=1}^n Q_Z(z_t)
=
Q_Z^{\otimes n}(z^n).
\]

Thus, the wiretapper's observation
has exactly the same distribution
for every identification message.

In particular, for all $i\neq j$,
\[
\left\|
P_{Z^n|i}
-
P_{Z^n|j}
\right\|_{\mathrm{TV}}
=0.
\]

Therefore, every non-secure
deterministic identification
feedback code automatically
satisfies the secrecy constraint.

It follows that
\[
C_{\mathrm{SID}}^d(D,W_S,V_S)
\geq
C_{\mathrm{ID}}^d(D,W_S).
\]

The reverse inequality follows
immediately because every secure
identification code is also a
non-secure identification code.

This proves the claim.
\end{proof}

\begin{remark}
The condition $H_{\mathrm{priv}}(D)>0$
is sufficient but not necessary for
achieving the non-secure deterministic
identification capacity.

To see this, consider the binary
state-dependent wiretap channel
\[
Y=X\oplus S,\qquad Z=S,
\]
where $S\sim\operatorname{Bern}(p)$,
$0<p<1/2$, and the state symbols
are independent across channel uses.

Let the distortion measure be Hamming
distortion and choose $D=0$.
Since $S=X\oplus Y$, the transmitter
can reconstruct the state perfectly
from its input and the feedback output.
Consequently,
\[
d_S^*(x)=0
\]
for every $x\in\{0,1\}$, and hence
$\mathcal X_D=\{0,1\}$.

The state-averaged main channel is
a binary symmetric channel with
crossover probability $p$. Therefore,
\[
C^d(0,W_S)=1-h_2(p)>0,
\]
where $h_2$ denotes the binary entropy
function.

Moreover,
\[
C_{\mathrm{ID}}^d(0,W_S)=h_2(p)>0.
\]

On the other hand, since $Y$ is
completely determined by $(X,Z)$,
\[
H(Y|Z,X=x)=0
\]
for every input symbol $x$, and thus
\[
H_{\mathrm{priv}}(0)=0.
\]

Nevertheless, the wiretapper observes
only the state sequence $Z^n=S^n$.
Since the states are i.i.d. and
independent of the identification
message, every deterministic feedback
encoding strategy induces
\[
P_{Z^n|i}=P_S^{\otimes n}
\]
for every identity $i$.

Consequently, for all $i\neq j$,
\[
\left\|
P_{Z^n|i}-P_{Z^n|j}
\right\|_{\mathrm{TV}}=0.
\]

Thus, every non-secure deterministic
identification feedback code for the
main channel automatically satisfies
the secrecy constraint.

Combining this observation with the
non-secure identification capacity
gives
\[
C_{\mathrm{SID}}^d(0,W_S,V_S)
=
C_{\mathrm{ID}}^d(0,W_S)
=
h_2(p)>0.
\]

This example demonstrates that
$H_{\mathrm{priv}}(D)>0$ is not a
necessary condition for positive
secure deterministic identification
capacity, nor for achieving the
full non-secure identification
capacity.
\end{remark}
    
    \begin{remark}
    Two achievability constructions are possible.
    The OTP approach uses a feedback-extracted
    secret key to encrypt the ID tag and requires
    $C^d(D,W_S)>0$ and $H_{\mathrm{priv}}(D)>0$.
    Alternatively, the extracted private randomness
    can seed a stochastic wiretap encoder when
    $C_S^d(D,W_S,V_S)>0$ and
    $H_{\mathrm{priv}}(D)>0$.
    Both approaches require privacy amplification,
    but the latter has stronger sufficient conditions.
    \end{remark} 

    The converse upper bound follows by removing
the secrecy constraint and applying the
non-secure identification converse with
feedback~\cite{labidi2024joint} to the
state-averaged main channel with input
alphabet restricted to $\mathcal X_D$.
Since every secure identification code
is also a non-secure identification code,
we obtain
\[
C_{\mathrm{SID}}^d(D,W_S,V_S)
\leq C_{\mathrm{ID}}^d(D,W_S).
\]
In particular, if
$C_{\mathrm{ID}}^d(D,W_S)=0$,
the secure identification capacity
must also be zero.
   The case $C^d(D,W_S)=0$ is already covered
by the preceding upper bound, since the
non-secure identification capacity formula
implies $C_{\mathrm{ID}}^d(D,W_S)=0$.

For completeness, we provide a direct
argument for this case. If
$C^d(D,W_S)=0$, all admissible input
symbols induce the same output
distribution, i.e.,
    \[
    W(\cdot|x)=W(\cdot|x'),
    \forall x,x'\in\mathcal{X}_D.
    \]
    Since the channel is memoryless, this implies
    that, for every identity $i$ and feedback history
    $y^{t-1}$,
    \begin{equation}
    P_{Y_t|Y^{t-1},i}(y_t|y^{t-1},i)
    =Q_Y(y_t),
    \end{equation}
    where $Q_Y$ is the common output distribution.
    Consequently,
    $P_{Y^n|i}=Q_Y^{\otimes n}$ for every $i$.
    
Thus, even with feedback, the receiver's
output distribution is independent of the
ID message. For any two distinct identities
$i$ and $j$, the first-kind error constraint
implies
\begin{equation}
Q_Y^{\otimes n}(\mathcal D_j)
\geq 1-\lambda,
\end{equation}
whereas the second-kind error constraint
implies
\begin{equation}
Q_Y^{\otimes n}(\mathcal D_j)
\leq\lambda.
\end{equation}
These inequalities are incompatible for
$\lambda<1/2$. Hence, reliable
identification is impossible, proving the
zero-capacity claim.

\section{Direct Part (Achievability) of Theorem~\ref{theorm:secureCapaciy}}


    Assume $C^d(D,W_S)>0$ and
    $H_{\mathrm{priv}}(D)>0$. Let
    $x^\star\in\arg\max_{x\in\mathcal X_D}
    H(Y|X=x)$.
    
    We emphasize that no local randomization is available at the encoder. Instead, secrecy will be achieved by exploiting the feedback-generated observations at the legitimate receiver, from which a small private component can be extracted and used to protect the message-dependent part of the protocol.
    The key idea is that, although the encoder is deterministic, the noiseless feedback induces a shared random observation at the legitimate receiver, which can be exploited as a common randomness resource. This feedback-generated randomness yields the full ID rate.
    Importantly, the key in Phase~III is generated separately from the common randomness $U_n$ used for identification, using feedback observations satisfying $H(Y|Z,X=x_0)>0$, which enables privacy amplification.
   As such, \(U_n\) need not be secret, and the input symbol \(x^\star\) used for common-randomness generation need not itself provide a positive secret-key rate.
    Instead, secrecy is ensured by extracting a small private component from the legitimate observations that remains unknown to the wiretapper and using it to protect a short message-dependent hash. Since this hash has sublinear length, a private key of vanishing rate is sufficient.

    \begin{remark}
    Note that the feedback-generated common random variable $U_n$
    used for identification is unrelated to the auxiliary random variable $U$
    appearing in the single-letter secrecy-capacity expression.
    \end{remark}
    
    
    Fix $R<C_{\mathrm{ID}}^d(D,W_S)$ and choose
    $\delta>0$ such that
    $R<H(Y|X=x^\star)-3\delta$.
    
    We construct a secure deterministic ID feedback
    code consisting of three successive parts. The first $n$
    channel uses generate the common randomness required for
    ID. A second, message-independent feedback phase
    of length $r_n=o(n)$ generates the secret key $K_n$. Finally,
    the protected hash is transmitted using
    $q_n=o(n)$ additional channel uses. Hence, the total
    blocklength is
    \[
    m=n+r_n+q_n=n+o(n).
    \]
    Since the two additional phases have sublinear length, they do not affect the asymptotic ID rate.

    \paragraph*{Phase I: Deterministic Common-Randomness Generation}
    During the first \(n\) channel uses, independently of the
    ID message, the encoder repeatedly transmits the fixed symbol
    \(x^\star\), where
    \[
    x^\star\in\arg\max_{x\in\mathcal X_D} H(Y|X=x).
    \]
    Hence,
    \[
    X_t=x^\star,\qquad t=1,\ldots,n,
    \]
    and \(Y^n\) is generated according to
    \(P_{Y|X=x^\star}^{\,n}\). Since the feedback is noiseless,
    the same realization \(Y^n\) is available to both the sender
    and the legitimate receiver.
    
    Since $X_t=x^\star$ throughout this phase,
    the output sequence is i.i.d. according to
    $P_{Y|X=x^\star}$. Choose $R_0$ such that
\[
R<R_0<H(Y|X=x^\star).
\]
Then a deterministic mapping of $Y^n$
yields a shared random variable $U_n$
with alphabet size $M_n$
    such that
    \begin{equation}
    \log M_n\geq nR_0-o(n)
    \end{equation}
    and
    \begin{equation}
    \left\|P_{U_n}
    -\operatorname{Unif}(\mathcal U_n)
    \right\|_{\mathrm{TV}}
    \leq\alpha_n,\qquad \alpha_n\to0.
    \end{equation}
    Here, $M_n=|\mathcal U_n|$, and
$\operatorname{Unif}(\mathcal U_n)$ denotes
the uniform distribution on $\mathcal U_n$.
    The mapping is fixed in advance and requires
    no local encoder randomization. This follows from the deterministic
    feedback-based common-randomness construction
    in~\cite{longversion}.

    Since \(x^\star\in\mathcal X_D\),
    \[
    d_S^*(x^\star)\le D,
    \]
    and thus the sensing distortion constraint is satisfied during
    this phase.

\paragraph*{Phase II: Identification from Common Randomness}

By construction, the long feedback phase is
message-independent and serves only to generate
the shared random variable $U_n$ and to satisfy
the sensing constraint.

Choose
\[
\ell_n=\lceil\sqrt n\rceil,
\qquad L_n=2^{\ell_n}.
\]
In particular, $\ell_n\to\infty$ and
$\ell_n=o(n)$.

We construct deterministic mappings
\begin{equation}
\tilde b_n(i,\cdot):
\mathcal U_n\to\{0,1\}^{\ell_n},
\qquad i=1,\ldots,N_n,
\end{equation}
using the probabilistic method.

For each identity $i$, independently
assign to every $u\in\mathcal U_n$
a uniformly distributed label from
$\{0,1\}^{\ell_n}$.

For any two distinct identities
$i\neq j$, the collision indicators
\[
\mathbf 1\{
\tilde b_n(i,u)=\tilde b_n(j,u)
\},
\qquad u\in\mathcal U_n,
\]
are independent Bernoulli random
variables with mean $1/L_n$.

Hence, Hoeffding's inequality gives
\begin{equation}
\Pr\left\{
\frac{1}{M_n}
\sum_{u\in\mathcal U_n}
\mathbf 1\{
\tilde b_n(i,u)=\tilde b_n(j,u)
\}
>
\frac{1}{L_n}+\frac{1}{\sqrt{L_n}}
\right\}
\leq
\exp\left(-\frac{2M_n}{L_n}\right).
\end{equation}

Choose
\[
N_n=
\left\lfloor
\exp\left(\frac{M_n}{4L_n}\right)
\right\rfloor.
\]
By the union bound, the probability
that at least one pair violates
the preceding collision bound
is at most
\[
N_n^2\exp\left(-\frac{2M_n}{L_n}\right)
\leq
\exp\left(-\frac{3M_n}{2L_n}\right),
\]
which tends to zero.

Consequently, for all sufficiently
large $n$, there exists a fixed
deterministic family of mappings
satisfying, for every $i\neq j$,
\[
\frac{1}{M_n}
\sum_{u\in\mathcal U_n}
\mathbf 1\{
\tilde b_n(i,u)=\tilde b_n(j,u)
\}
\leq
\delta_n,
\]
where
\[
\delta_n=
\frac{1}{L_n}+\frac{1}{\sqrt{L_n}}
\longrightarrow 0.
\]

Since
\[
\left\|
P_{U_n}-\operatorname{Unif}(\mathcal U_n)
\right\|_{\mathrm{TV}}
\leq\alpha_n,
\]
we obtain
\begin{equation}
\Pr\left\{
\tilde b_n(i,U_n)=\tilde b_n(j,U_n)
\right\}
\leq\delta_n+\alpha_n.
\end{equation}

Furthermore, since
$\log M_n\geq nR_0-o(n)$
and $\ell_n=o(n)$, our choice of
$N_n$ implies
\[
\liminf_{n\to\infty}
\frac{1}{n}\log\log N_n
\geq R_0.
\]
As $R_0$ can be chosen arbitrarily
close to $H(Y|X=x^\star)$,
every identification rate
\[
R<H(Y|X=x^\star)
\]
is achievable.

The receiver tests a candidate
identity $j$ by comparing the
recovered tag with
$\tilde b_n(j,U_n)$.

\paragraph*{Phase III: Feedback-Based Secret-Key Generation and
Hash Protection}
Let
\[
\tilde b_n(i,U_n)\in\{0,1\}^{\ell_n}
\]
denote the sublinear message-dependent hash obtained in
Phase~II.
Consider an additional message-independent feedback phase of
length \(r_n=o(n)\), where the transmitter repeatedly sends a
fixed input symbol \(x_0\in\mathcal X_D\) satisfying
\[
H(Y|Z,X=x_0)>0.
\]

Such an input symbol $x_0$ exists since
$H_{\mathrm{priv}}(D)>0$. Let
$h_0:=H(Y|Z,X=x_0)>0$.

For the tag length $\ell_n$ chosen in
Phase II, choose $r_n\to\infty$ with
$r_n=o(n)$ such that, for some
$0<\epsilon<h_0$,
\begin{equation}
\ell_n<r_n(h_0-\epsilon).
\end{equation}

The legitimate receiver observes \(\bar{Y}^{r_n}\), and noiseless
feedback makes the same sequence available to the transmitter,
whereas the wiretapper observes \(\bar{Z}^{r_n}\).


Since the input symbol $x_0$ is fixed,
the pairs $(\bar Y_t,\bar Z_t)$ are
i.i.d. according to $P_{YZ|X=x_0}$.

Let $\epsilon>0$ be the constant chosen
above such that
\[
\ell_n<r_n(h_0-\epsilon).
\]
Define the event
\[
\mathcal G_n
=
\left\{
-\log
P_{\bar Y^{r_n}|\bar Z^{r_n}}
(\bar Y^{r_n}|\bar Z^{r_n})
\geq r_n(h_0-\epsilon/2)
\right\}.
\]
By the conditional asymptotic
equipartition property,
\[
\beta_n:=\Pr(\mathcal G_n^c)
\longrightarrow 0.
\]

Consider a two-universal family of
hash functions from
$\mathcal Y^{r_n}$ to
$\{0,1\}^{\ell_n}$.

By privacy amplification
~\cite{bennett1995generalized},
averaging over this family yields
\begin{equation}
\begin{aligned}
&\mathbb E_F
\left[
\left\|
P_{F(\bar Y^{r_n}),\bar Z^{r_n}|F}
-
U_{\ell_n}P_{\bar Z^{r_n}}
\right\|_{\mathrm{TV}}
\right]
\\
&\quad\leq
\beta_n+
\frac12
2^{-\frac12[
r_n(h_0-\epsilon/2)-\ell_n]}.
\end{aligned}
\end{equation}

Since
$\ell_n<r_n(h_0-\epsilon)$,
we obtain
\[
\beta_n+
\frac12
2^{-\frac12[
r_n(h_0-\epsilon/2)-\ell_n]}
\leq
\beta_n+\frac12 2^{-r_n\epsilon/4}
\longrightarrow 0.
\]

Consequently, there exists a fixed
deterministic mapping
\[
g_n:\mathcal Y^{r_n}
\to\{0,1\}^{\ell_n}
\]
such that
\[
K_n=g_n(\bar Y^{r_n})
\]
satisfies
\begin{equation}
\left\|
P_{K_n,\bar Z^{r_n}}
-
U_{\ell_n}P_{\bar Z^{r_n}}
\right\|_{\mathrm{TV}}
\leq\eta_n,
\end{equation}
where
\[
\eta_n
:=
\beta_n+\frac12 2^{-r_n\epsilon/4}
\longrightarrow 0.
\]

The mapping $g_n$ is selected during
code design and is fixed and publicly
known before transmission.
Therefore, its implementation requires
no local randomization at the encoder.

Moreover, since the common-randomness
generation phase and the key-generation
phase use fixed input symbols over
disjoint channel uses, the memoryless
channel property implies that
$(K_n,\bar Z^{r_n})$ is independent of
$(U_n,Z^n)$.
Since both legitimate terminals observe the same
\(\bar{Y}^{r_n}\), they compute the same key \(K_n\). Since
\(g_n\) is deterministic, no local encoder randomization is required.
The identification hash \(\tilde b_n\) is now protected using a
one-time pad with the feedback-extracted key \(K_n\). Specifically,
the sender computes
\[
\hat b_n
=
\tilde b_n(i,U_n)\oplus K_n.
\]
The protected hash $\hat b_n$ is transmitted
using a deterministic transmission code
over the state-averaged main channel with
input alphabet $\mathcal X_D$.

Since $C^d(D,W_S)>0$, fix any rate
\[
0<R_{\mathrm{tag}}<C^d(D,W_S)
\]
and choose
\[
q_n=
\left\lceil
\frac{\ell_n}{R_{\mathrm{tag}}}
\right\rceil.
\]
Then
\[
\frac{\ell_n}{q_n}
\leq R_{\mathrm{tag}}
<C^d(D,W_S).
\]

By the channel coding theorem for finite
discrete memoryless channels, there exists
a sequence of deterministic transmission
codes with $2^{\ell_n}$ messages and
maximal decoding error probability
\[
\epsilon_n
:=
\max_{b\in\{0,1\}^{\ell_n}}
\Pr\{B_n^{\mathrm{dec}}\neq b
\mid \hat b_n=b\}
\longrightarrow 0.
\]

All codewords belong to
$\mathcal X_D^{q_n}$ and therefore
satisfy the symbol-wise sensing constraint.

Since $\ell_n\to\infty$ and
$\ell_n=o(n)$, we have
$q_n\to\infty$ and $q_n=o(n)$.
Thus, the reliable transmission of the
encrypted tag requires only a sublinear
number of additional channel uses.
The legitimate receiver first computes
$K_n$ from $\bar Y^{r_n}$.

Let $B_n^{\mathrm{dec}}$ denote the
decoder's estimate of the transmitted
encrypted tag $\hat b_n$.

The receiver then computes
\[
\widehat T_n
=
B_n^{\mathrm{dec}}\oplus K_n.
\]

For each candidate identity $j$, the
receiver performs the deterministic
identification test
\[
\widehat T_n
\stackrel{?}{=}
\tilde b_n(j,U_n).
\]

The receiver accepts identity $j$
if the equality holds and rejects it
otherwise.
    
\paragraph*{Reliability Analysis}
The reliability analysis involves the following
contributions:
\begin{enumerate}
    \item The approximation error $\alpha_n$
    associated with the common-randomness
    distribution in Phase I;

    \item The collision probability of the
    deterministic identification coloring
    construction based on $U_n$ and
    $\tilde b_n$;

    \item The decoding error of the reliable
    code used to transmit the protected
    hash $\hat b_n$.
\end{enumerate}

    
Let $\alpha_n$ denote the common-randomness
approximation error from Phase I,
$\delta_n$ the collision bound of the
coloring construction, and $\epsilon_n$
the maximal decoding error of the reliable
tag-transmission code.

For the resulting identification code
of total blocklength
$m=n+r_n+q_n$, define the maximal
error probabilities of the first and
second kind by
\begin{align}
\lambda_{1,n}
&:=
\max_{1\le i\le N_n}
P_{Y^m|i}(\mathcal D_i^c),
\\
\lambda_{2,n}
&:=
\max_{\substack{1\le i,j\le N_n\\i\neq j}}
P_{Y^m|i}(\mathcal D_j),
\end{align}
where $\mathcal D_i\subseteq\mathcal Y^m$
denotes the decoding region associated
with identity $i$.

These definitions are consistent with
the reliability requirements in
Definition~\ref{def:1}.

Since both legitimate terminals obtain the
same feedback sequence, they agree on
$U_n$ and $K_n$ without error. Therefore,
the first-kind ID error satisfies
\begin{equation}
\lambda_{1,n}\leq\epsilon_n,
\end{equation}
while the second-kind error is bounded by
\begin{equation}
\lambda_{2,n}
\leq\delta_n+\alpha_n+\epsilon_n.
\end{equation}
As $\alpha_n,\delta_n,\epsilon_n\to0$,
both error probabilities vanish.




\paragraph*{Security Analysis}

Both the common-randomness and key-generation
phases are independent of the ID message.
Since these phases use fixed channel inputs
over disjoint channel uses, the pair
$(K_n,\bar Z^{r_n})$ is independent of
$(U_n,Z^n)$.

Let $K_n'$ be an independent uniform key
on $\{0,1\}^{\ell_n}$. For identity $i$,
define the ideal encrypted tag
\begin{equation}
\hat b_n'
=
\tilde b_n(i,U_n)\oplus K_n'.
\end{equation}
Conditioned on $U_n$ and the wiretapper's
observations from both preliminary phases,
$\hat b_n'$ is uniformly distributed over
$\{0,1\}^{\ell_n}$, independently of
the identity $i$.

Since $(K_n,\bar Z^{r_n})$ is independent
of $(U_n,Z^n)$, the privacy-amplification
bound implies
\begin{equation}
\begin{aligned}
&\bigl\|
P_{U_n,Z^n,\bar Z^{r_n},K_n}\\
&\quad -
P_{U_n,Z^n,\bar Z^{r_n}}U_{\ell_n}
\bigr\|_{\mathrm{TV}}
\leq \eta_n.
\end{aligned}
\end{equation}
Thus, replacing the actual key by an
independent uniform key changes the
joint distribution of all preliminary
observations and the key by at most
$\eta_n$ in total variation.

The encrypted tag is subsequently transmitted
using the same deterministic channel code
in both constructions. Since total
variation distance cannot increase under
a common channel, the wiretapper's complete
observation under identity $i$ differs
from its ideal counterpart by at most
$\eta_n$.

In the ideal construction, the distribution
of the wiretapper's complete observation
does not depend on the identity.
Consequently, for every $i\neq j$,
the triangle inequality yields
\begin{equation}
\left\|
P_{Z^m|i}-P_{Z^m|j}
\right\|_{\mathrm{TV}}
\leq 2\eta_n\longrightarrow 0.
\end{equation}
Thus, the secrecy requirement is satisfied.

\paragraph*{Distortion Analysis}
Let
\[
d_t:=\mathbb E[d(S_t,\hat S_t)]
\]
denote the expected distortion at time $t$.

During Phase~I, the encoder transmits
\(x^\star\in\mathcal X_D\).
Hence,
\[
d_t=d_S^*(x^\star)\le D
\]
for every channel use in this phase. During the secret-key
generation phase, the encoder transmits \(x_0\in\mathcal X_D\),
and therefore
\[
d_t=d_S^*(x_0)\le D.
\]
The protected hash is transmitted using codewords in
\(\mathcal X_D^{q_n}\), so \(d_t\le D\) also in this phase.
Hence, \(d_t\le D\) for all \(t=1,\ldots,m\), and the
per-symbol sensing constraint is satisfied throughout the
entire scheme.
\paragraph*{Rate Computation and Conclusion}
Since $r_n=o(n)$ and $q_n=o(n)$, the total
blocklength is $m=n+r_n+q_n=n+o(n)$.
Hence, $n/m\to1$. By the coloring construction
in Phase II, every rate
$R<C_{\mathrm{ID}}^d(D,W_S)$ is achievable, i.e.,
\begin{equation}
\liminf_{n\to\infty}
\frac{1}{m}\log\log N_n\geq R.
\end{equation}
Since $R<C_{\mathrm{ID}}^d(D,W_S)$ is arbitrary,
\begin{equation}
C_{\mathrm{SID}}^d(D,W_S,V_S)
\geq C_{\mathrm{ID}}^d(D,W_S)
\end{equation}
whenever $C^d(D,W_S)>0$ and
$H_{\mathrm{priv}}(D)>0$.
Together with the converse upper bound,
\begin{equation}
C_{\mathrm{SID}}^d(D,W_S,V_S)
=C_{\mathrm{ID}}^d(D,W_S),
\end{equation}
which completes the achievability proof of
Theorem~\ref{theorm:secureCapaciy}.

\section{Alternative Achievability via Stochastic Wiretap Encoding}

We present an alternative achievability
construction based on stochastic wiretap
encoding. Unlike the one-time-pad approach
in the previous section, the private randomness
extracted from feedback is used to simulate
the randomization required by a stochastic
wiretap encoder. This construction applies
under the stronger condition of positive
admissible wiretap secrecy capacity and
provides a second route to achieving the
non-secure identification capacity without
local encoder randomness.

Assume $C_S^d(D,W_S,V_S)>0$ and
$H_{\mathrm{priv}}(D)>0$.

Recall that the classical wiretap secrecy
capacity cannot exceed the ordinary
transmission capacity of the main channel.
Indeed, for every auxiliary random variable
$U$ satisfying $U-X-(Y,Z)$,
\[
I(U;Y)-I(U;Z)
\leq I(U;Y)
\leq I(X;Y).
\]
Consequently,
\[
C_S^d(D,W_S,V_S)
\leq C^d(D,W_S).
\]
Hence, the assumption
$C_S^d(D,W_S,V_S)>0$
also guarantees
$C^d(D,W_S)>0$.

The alternative construction therefore
operates under stronger assumptions
than the one-time-pad construction
presented in the previous section.

Fix any
$R<C_{\mathrm{ID}}^d(D,W_S)$.
As in Phase I of the OTP construction,
transmit a fixed symbol
$x^\star\in\mathcal X_D$ maximizing
$H(Y|X=x)$ for $n$ channel uses.
Noiseless feedback provides both legitimate
terminals with common randomness $U_n$
at any rate below $C_{\mathrm{ID}}^d(D,W_S)$.
This phase is independent of the ID message,
and $U_n$ need not be secret.

Fix any rate $R_s$ satisfying
\[
0<R_s<C_S^d(D,W_S,V_S).
\]

For the state-averaged discrete memoryless
wiretap channel with input alphabet
restricted to $\mathcal X_D$, the
semantic-security wiretap coding theorem
\cite{yang2019wiretap} guarantees the
existence of stochastic wiretap codes
with blocklength $q_n\to\infty$,
$q_n=o(n)$, and tag alphabet $\mathcal T_n$
satisfying
\[
\log|\mathcal T_n|
\geq q_nR_s-o(q_n).
\]

These codes can be chosen to satisfy
both maximal decoding reliability
and semantic security. More precisely,
there exist sequences
$\rho_n\to0$ and $\xi_n\to0$
such that
\begin{equation}
\max_{t\in\mathcal T_n}
\Pr\{\hat T\neq t\mid T=t\}
\leq\rho_n
\end{equation}
and
\begin{equation}
\sup_{P_T}I(T;Z^{q_n})
\leq\xi_n.
\end{equation}

Here, $P_{Z^{q_n}|t}$ denotes the
wiretapper's output distribution
induced by the stochastic encoder
when tag $t$ is transmitted, averaged
over the encoder's independent
randomization.

The semantic-security condition
implies the uniform total-variation
bound needed in our construction.
To see this, fix an arbitrary
$t_0\in\mathcal T_n$ and define
\[
Q_n:=P_{Z^{q_n}|t_0}.
\]

For any $t\neq t_0$, consider a
binary random variable $T$ uniformly
distributed over $\{t,t_0\}$.
By semantic security,
\[
I(T;Z^{q_n})\leq\xi_n.
\]

Applying Pinsker's inequality to
the two conditional output
distributions gives
\begin{equation}
\left\|
P_{Z^{q_n}|t}
-
P_{Z^{q_n}|t_0}
\right\|_{\mathrm{TV}}
\leq
\sqrt{2(\ln 2)\xi_n}.
\end{equation}

Consequently,
\begin{equation}
\max_{t\in\mathcal T_n}
\left\|
P_{Z^{q_n}|t}-Q_n
\right\|_{\mathrm{TV}}
\leq
\sqrt{2(\ln 2)\xi_n}
\longrightarrow 0.
\end{equation}

Thus, by defining
\[
\epsilon_n
:=
\max\left\{
\rho_n,\sqrt{2(\ln 2)\xi_n}
\right\},
\]
both the maximal decoding error
and the uniform total-variation
secrecy bound are upper bounded
by $\epsilon_n\to0$.

Since $R_s>0$, we also have
$|\mathcal T_n|\to\infty$,
which permits the identification
coloring construction used below.
Since $\mathcal X_D$ is finite, the
stochastic encoder for each tag $t$
induces a probability distribution
$P_{X^{q_n}|t}$ on $\mathcal X_D^{q_n}$.

Let
\[
k_n=
\left\lceil
q_n\log_2|\mathcal X_D|
+2\log_2 q_n
\right\rceil.
\]
For every tag $t$, we can approximate
$P_{X^{q_n}|t}$ by a distribution
$\widetilde P_{X^{q_n}|t}$
whose probabilities are integer
multiples of $2^{-k_n}$.

Indeed, rounding the probabilities
to multiples of $2^{-k_n}$ and
redistributing the remaining mass gives
\[
\left\|
P_{X^{q_n}|t}
-\widetilde P_{X^{q_n}|t}
\right\|_{\mathrm{TV}}
\leq
\frac{|\mathcal X_D|^{q_n}}{2^{k_n}}
\leq q_n^{-2}.
\]

Consequently, there exists a
deterministic mapping
\[
\phi_n:
\mathcal T_n\times\{0,1\}^{k_n}
\to\mathcal X_D^{q_n}
\]
such that, for an independent uniform
seed $K_n'$ on $\{0,1\}^{k_n}$,
the conditional distribution of
$\phi_n(t,K_n')$ equals
$\widetilde P_{X^{q_n}|t}$.

By the data-processing inequality for
total variation, this approximation
changes both the decoding-error
probability and the wiretapper's
output distribution by at most
$q_n^{-2}$, uniformly over all tags.

Thus, after replacing $\epsilon_n$
by $\epsilon_n+q_n^{-2}$ and denoting
the resulting vanishing error bound
again by $\epsilon_n$, the maximal
reliability and uniform secrecy
guarantees remain valid.

In particular, $k_n=O(q_n)$.

In the proposed feedback construction,
this seed is not generated locally.
Instead, it is extracted from the
legitimate channel outputs in a separate
message-independent phase. The
privacy-amplification mapping is selected
during code design and fixed before
transmission.

Choose $x_0\in\mathcal X_D$ satisfying
$h_0:=H(Y|Z,X=x_0)>0$.
Transmit $x_0$ during a separate
message-independent phase of length
$r_n\to\infty$, with $r_n=o(n)$,
chosen such that, for some
$0<\delta<h_0$,
\begin{equation}
k_n<r_n(h_0-\delta).
\end{equation}
Since
\[
\lim_{n\to\infty}\frac{k_n}{q_n}
=\log_2|\mathcal X_D|,
\]
we may fix a constant
\[
c>\frac{\log_2|\mathcal X_D|}{h_0-\delta}
\]
and choose
\[
r_n=\lceil c q_n\rceil.
\]
Then, for all sufficiently large $n$,
\[
k_n<r_n(h_0-\delta).
\]
Since $q_n\to\infty$ and $q_n=o(n)$,
we also have $r_n\to\infty$ and
$r_n=o(n)$.

Consequently, the total blocklength
satisfies
\[
m=n+r_n+q_n=n+o(n).
\]
Thus,
\[
\liminf_{n\to\infty}
\frac{1}{m}\log\log N_n
\geq R,
\]
and the additional phases do not reduce
the asymptotically achievable ID rate.

By privacy amplification
\cite{bennett1995generalized},
there exists a fixed deterministic mapping
$g_n:\mathcal Y^{r_n}\to\{0,1\}^{k_n}$
such that $K_n=g_n(\bar Y^{r_n})$
satisfies
\begin{equation}
\left\|
P_{K_n,\bar Z^{r_n}}
-U_{k_n}P_{\bar Z^{r_n}}
\right\|_{\mathrm{TV}}
\leq\eta_n,\qquad\eta_n\to0.
\end{equation}
Since this phase uses fixed inputs and
the channel is memoryless, the extracted
seed and its corresponding wiretapper
observations are independent of Phase I.

The mapping $g_n$ is selected and fixed
before transmission. Consequently,
$K_n$ is generated solely from channel
feedback, without local randomization
at the encoder.

We now construct the identification
coloring functions.

Fix $R_0$ such that
\[
R<R_0<H(Y|X=x^\star).
\]
As in Phase I, the common random variable
$U_n$ has alphabet size $M_n$ satisfying
\[
\log M_n\geq nR_0-o(n)
\]
and
\[
\left\|
P_{U_n}-\operatorname{Unif}(\mathcal U_n)
\right\|_{\mathrm{TV}}
\leq\alpha_n,
\qquad \alpha_n\to0.
\]

Let $\mathcal T_n^{\mathrm{original}}$
denote the tag alphabet of the
stochastic wiretap code constructed above.

Since
\[
\log|\mathcal T_n^{\mathrm{original}}|
\geq q_nR_s-o(q_n),
\]
we may choose a subset
$\mathcal T_n\subseteq
\mathcal T_n^{\mathrm{original}}$
of size
\[
L_n:=|\mathcal T_n|
=\left\lfloor 2^{q_nR_s/2}\right\rfloor
\]
for all sufficiently large $n$.

Restricting the wiretap code to this
subset preserves its maximal reliability
and uniform secrecy guarantees.
In particular,
\[
L_n\to\infty,
\qquad
\log L_n=o(n).
\]

For every identity $i$ and every
$u\in\mathcal U_n$, independently
choose $h_i(u)$ uniformly from
$\mathcal T_n$.

For two distinct identities $i\neq j$,
the random variables
\[
\mathbf 1\{h_i(u)=h_j(u)\},
\qquad u\in\mathcal U_n,
\]
are independent Bernoulli variables
with mean $1/L_n$.

By Hoeffding's inequality,
\begin{equation}
\Pr\left\{
\frac{1}{M_n}
\sum_{u\in\mathcal U_n}
\mathbf 1\{h_i(u)=h_j(u)\}
>
\frac{1}{L_n}+\frac{1}{\sqrt{L_n}}
\right\}
\leq
\exp\left(-\frac{2M_n}{L_n}\right).
\end{equation}

Choose
\[
N_n=
\left\lfloor
\exp\left(\frac{M_n}{4L_n}\right)
\right\rfloor.
\]

By the union bound, the probability
that any pair violates the collision
bound is at most
\[
N_n^2
\exp\left(-\frac{2M_n}{L_n}\right)
\leq
\exp\left(-\frac{3M_n}{2L_n}\right),
\]
which tends to zero.

Consequently, there exist fixed
deterministic mappings
\[
h_i:\mathcal U_n\to\mathcal T_n,
\qquad i=1,\ldots,N_n,
\]
such that, for every $i\neq j$,
\[
\Pr\{h_i(U_n)=h_j(U_n)\}
\leq\delta_n+\alpha_n,
\]
where
\[
\delta_n=
\frac{1}{L_n}+\frac{1}{\sqrt{L_n}}
\longrightarrow0.
\]

Finally,
\[
\log\log N_n
=
\log M_n-\log L_n+O(1).
\]
Since $\log L_n=o(n)$, we obtain
\[
\liminf_{n\to\infty}
\frac{1}{n}\log\log N_n
\geq R_0>R.
\]

Thus, the construction supports
doubly exponentially many identities
with vanishing pairwise collision
probability.

For identity $i$, the transmitter sends
the tag $h_i(U_n)$ using the stochastic
wiretap encoder with the extracted seed
$K_n$. The receiver decodes the tag and
accepts a candidate identity $j$ if the
decoded tag equals $h_j(U_n)$.

\paragraph*{Reliability Analysis}

For each transmitted identity $i$, let
$\widehat T_n$ denote the decoded tag
and define the tag-decoding error event
\[
\mathcal E_i
=
\{\widehat T_n\neq h_i(U_n)\}.
\]

Consider first the ideal construction
in which the stochastic wiretap encoder
uses an independent uniform seed.

By the maximal decoding-error guarantee
of the wiretap code, we have
\[
P_i^{\mathrm{ideal}}(\mathcal E_i)
\leq\epsilon_n
\]
uniformly over all identities $i$.

In the actual construction, the
independent uniform seed is replaced
by the feedback-extracted seed $K_n$.

Since
\[
\left\|
P_{K_n,\bar Z^{r_n}}
-
U_{k_n}P_{\bar Z^{r_n}}
\right\|_{\mathrm{TV}}
\leq\eta_n
\]
and the seed-generation phase is
independent of Phase I, the
data-processing inequality implies
\[
\left|
P_i^{\mathrm{actual}}(\mathcal E_i)
-
P_i^{\mathrm{ideal}}(\mathcal E_i)
\right|
\leq\eta_n.
\]

Consequently,
\[
P_i^{\mathrm{actual}}(\mathcal E_i)
\leq\epsilon_n+\eta_n
\]
uniformly over all identities $i$.

Since the receiver accepts identity $i$
whenever the transmitted tag is decoded
correctly, the maximal first-kind
identification error satisfies
\begin{equation}
\lambda_{1,n}
\leq\epsilon_n+\eta_n.
\end{equation}

For the second-kind error, consider
two distinct identities $i\neq j$.
The receiver can incorrectly accept
identity $j$ only if either the
transmitted tag is decoded incorrectly
or the two coloring functions collide.

Therefore,
\begin{align}
&P_i^{\mathrm{actual}}
\{\widehat T_n=h_j(U_n)\}
\\
&\quad\leq
P_i^{\mathrm{actual}}(\mathcal E_i)
+
\Pr\{h_i(U_n)=h_j(U_n)\}
\\
&\quad\leq
\epsilon_n+\eta_n+\delta_n+\alpha_n.
\end{align}

Taking the maximum over all distinct
identity pairs gives
\begin{equation}
\lambda_{2,n}
\leq
\delta_n+\alpha_n+\epsilon_n+\eta_n.
\end{equation}

Since
$\alpha_n,\delta_n,\epsilon_n,\eta_n
\to0$,
both maximal identification error
probabilities vanish.

For secrecy, let
\[
V_n=(Z^n,\bar Z^{r_n})
\]
denote the wiretapper's observations
during the two preliminary phases, and
let $Z_{\mathrm f}^{q_n}$ denote its
observation during the final phase.

Since both preliminary phases use fixed
channel inputs, the distribution $P_{V_n}$
is independent of the identification
message. Moreover, the memoryless channel
property implies that
$(K_n,\bar Z^{r_n})$ is independent of
$(U_n,Z^n)$.

Consider first an ideal construction in
which the stochastic wiretap encoder uses
an independent uniform $k_n$-bit seed.
Denote the corresponding distribution
by $P^{\mathrm{ideal}}$.

For any identity $i$, the transmitted tag
is $h_i(U_n)$. Although this tag may be
correlated with $Z^n$, the wiretap code
satisfies the uniform secrecy bound
\[
\max_{t\in\mathcal T_n}
\left\|
P_{Z^{q_n}|t}-Q_n
\right\|_{\mathrm{TV}}
\leq\epsilon_n.
\]
Hence, conditioning on $U_n$ and $V_n$
and using the convexity of total variation,
we obtain
\begin{equation}
\left\|
P^{\mathrm{ideal}}_{V_n,Z_{\mathrm f}^{q_n}|i}
-
P_{V_n}Q_n
\right\|_{\mathrm{TV}}
\leq\epsilon_n.
\end{equation}

Next, consider the actual construction
using the feedback-extracted seed $K_n$.

Since the two preliminary phases use
fixed channel inputs over disjoint
channel uses, the memoryless channel
property implies
\[
P_{U_n,Z^n,\bar Z^{r_n},K_n}
=
P_{U_n,Z^n}
P_{\bar Z^{r_n},K_n}.
\]

Moreover, the reference distribution
with an independent uniform seed
factorizes as
\[
P_{U_n,V_n}U_{k_n}
=
P_{U_n,Z^n}
P_{\bar Z^{r_n}}U_{k_n}.
\]

Consequently, by the invariance of
total variation under tensoring with
a common probability distribution,
\begin{align}
&\left\|
P_{U_n,V_n,K_n}
-
P_{U_n,V_n}U_{k_n}
\right\|_{\mathrm{TV}}
\\
&\quad =
\left\|
P_{\bar Z^{r_n},K_n}
-
P_{\bar Z^{r_n}}U_{k_n}
\right\|_{\mathrm{TV}}
\\
&\quad \leq \eta_n.
\end{align}

Thus, the extracted seed is
asymptotically indistinguishable
from an independent uniform seed,
even when the wiretapper's
observations from both preliminary
phases are taken into account.
Since the final-phase encoder and channel
are identical in both constructions,
the data-processing inequality yields
\begin{equation}
\left\|
P^{\mathrm{actual}}_{V_n,Z_{\mathrm f}^{q_n}|i}
-
P^{\mathrm{ideal}}_{V_n,Z_{\mathrm f}^{q_n}|i}
\right\|_{\mathrm{TV}}
\leq\eta_n.
\end{equation}

Consequently, for every identity $i$,
\begin{equation}
\left\|
P^{\mathrm{actual}}_{V_n,Z_{\mathrm f}^{q_n}|i}
-
P_{V_n}Q_n
\right\|_{\mathrm{TV}}
\leq\epsilon_n+\eta_n.
\end{equation}

Since the reference distribution
$P_{V_n}Q_n$ is independent of the
identity, the triangle inequality gives
\begin{equation}
\max_{i\neq j}
\left\|
P_{Z^m|i}-P_{Z^m|j}
\right\|_{\mathrm{TV}}
\leq 2(\epsilon_n+\eta_n)
\longrightarrow 0,
\end{equation}
where $m=n+r_n+q_n$.

Thus, the pairwise total-variation
secrecy requirement is satisfied.
All transmitted symbols belong to
$\mathcal X_D$, so the sensing constraint
is satisfied. Moreover, $r_n+q_n=o(n)$,
and hence $m=n+o(n)$. Thus, every
$R<C_{\mathrm{ID}}^d(D,W_S)$ is achievable
under the stronger conditions
$C_S^d(D,W_S,V_S)>0$ and
$H_{\mathrm{priv}}(D)>0$.

\section{Conclusion and Future Work}

We studied secure deterministic identification
with noiseless feedback over state-dependent
wiretap channels under a sensing distortion
constraint. The encoder is fully deterministic
and uses no local randomization; its actions
depend only on the identification message
and previously observed feedback.

We showed that the secure deterministic
identification capacity equals its non-secure
counterpart whenever the admissible main
channel has positive ordinary transmission
capacity and the private feedback entropy
is positive. Conversely, the secure
identification capacity is zero if either
the admissible main channel has zero
ordinary transmission capacity or the
non-secure identification capacity is zero.

Our primary achievability construction uses
feedback to generate common randomness for
identification and a sublinear secret key
for one-time-pad protection of the
identification tag. We also presented an
alternative construction based on stochastic
wiretap encoding under the conditions of
positive admissible wiretap secrecy capacity
and positive private feedback entropy. In this alternative, feedback-derived
private randomness replaces the local
randomness normally required by the
stochastic encoder.

The case of positive ordinary transmission
capacity and positive non-secure identification
capacity, but zero private feedback entropy,
is not completely characterized by our results.

Nevertheless, we have demonstrated that positive
private feedback entropy is not a necessary
condition for achieving the full non-secure
identification capacity. In particular, our
binary state-dependent wiretap channel example
shows that secure deterministic identification
with feedback can achieve a strictly positive
rate even when the private feedback entropy
vanishes.

This observation suggests that the existence
of private randomness in the feedback output
is not the only mechanism that can guarantee
secure identification. A complete characterization
of the secure deterministic identification
capacity in the absence of private feedback
entropy remains an interesting direction for
future research.

\section*{Acknowledgment}
The authors were supported in part by the German Federal Ministry of Research, Technology and Space (BMFTR) through the Transfer Hub 6G-life under Grants 16KIS2413K, 16KIS2414, and 16KIS2415. Holger Boche was additionally supported in part by the BMFTR within the national initiative on Molecular Communication (IoBNT) under Grant 16KIS1988. Christian Deppe and Vida Gholamian were further supported in part by the BMFTR within the national initiative QSTARS under Grant 16KIS2602. Christian Deppe, Vida Gholamian, and Yaning Zhao were also supported in part by the German Research Foundation (DFG) under Project DE1915/2-1. Holger Boche and Rafael F. Schaefer were further supported in part by the DFG as part of Germany’s Excellence Strategy--Excellence Cluster (EXC) 2050/2 - Project ID 390696704 - Cluster of Excellence \emph{``Centre for Tactile Internet with Human-in-the-Loop'' (CeTI)}, and by the DFG within the Priority Program SPP 2378 \emph{``Resilient Worlds''}, Project ID 565653523.

\bibliographystyle{IEEEtranS}
\bibliography{references}

\end{document}